\documentclass{IEEEtran}

\usepackage{amssymb}
\usepackage{amsthm}
\usepackage{cite}
\usepackage[multiple]{footmisc}
\usepackage{mathtools}
\usepackage{enumitem}

\newcommand{\mbb}{\mathbb}

\newcommand{\R}{\mathbb{R}}
\newcommand{\N}{\mathbb{N}}

\newcommand{\BR}{\mbox{\textup{BR}}}

   \def\vx{{\bf x}}

  \def\calL{\mathcal{L}}

\newtheorem{theorem}{Theorem}
\newtheorem{remark}[theorem]{Remark}
\newtheorem{mydef}[theorem]{Definition}
\newtheorem{lemma}[theorem]{Lemma}

\title{On the Exponential Rate of Convergence of Fictitious Play in Potential Games}
\author{Brian Swenson  and  Soummya Kar \thanks{The authors are with the Department of Electrical and Computer Engineering, Carnegie Mellon University, Pittsburgh, PA, USA (brianswe@ece.cmu.edu, soummyak@andrew.cmu.edu). This work was supported in part by the National Science Foundation under grants CCF:1513936 and CCF:1646526.}}

\begin{document}

\maketitle

\begin{abstract}
The paper studies fictitious play (FP) learning dynamics in continuous time. It is shown that in almost every potential game, and for almost every initial condition, the rate of convergence of FP is exponential.
In particular, the paper focuses on studying the behavior of FP in potential games in which all equilibria of the game are regular, as introduced by Harsanyi. Such games are referred to as \emph{regular potential games}.
Recently it has been shown that almost all potential games (in the sense of the Lebesgue measure) are regular. In this paper it is shown that in any regular potential game (and hence, in almost every potential game), FP converges to the set of Nash equilibria at an exponential rate from almost every initial condition.
\end{abstract}

\begin{IEEEkeywords}
  Game theory, Learning, Potential games, Fictitious play, Multi-agent systems, Best-response dynamics
\end{IEEEkeywords}

\section{Introduction}
Multi-agent systems are naturally modeled using the mathematical framework of game theory \cite{fudenberg1991game}.  In recent years there has been a surge of research investigating the use of game-theoretic learning processes as a means of controlling multi-agent systems in a decentralized manner (see \cite{marden2012game} and references therein). Application domains include, but are not limited to, wireless networks \cite{lasaulce2011game,rose2011learning}; the smart grid infrastructure \cite{saad2012game}; distributed traffic routing \cite{marden06,Lambert01,swenson2015CESFP}; electric vehicle charging networks \cite{ma2013decentralized}; mobile sensor networks \cite{zhu2013distributed}; and wind farm management \cite{marden2013model,wu2011demand}.
From an engineering perspective, if game-theoretic learning processes are to be used as decentralized control algorithms in such systems, then it is of paramount importance to understand the rate at which such processes converge to equilibrium.

One of the best-known and most prototypical game theoretic learning processes is known as \emph{fictitious play} (FP) \cite{Brown51,fudenberg1998theory}. In FP, each player adjusts their personal strategy towards a myopic best response given the current strategy of opponents. (See Section \ref{sec_prelims} for a formal definition.) Since Nash equilibria (NE) are defined as the fixed points of the best response mapping, FP dynamics can be regarded as the ``natural'' learning dynamics associated with the NE concept \cite{candogan2013dynamics,swenson2017FP_pot_games}.

Despite its prototypical role, there are relatively few rigorous results characterizing the rate of convergence of FP \cite{harris1998rate,brandt2010rate,daskalakis2014counter}. In particular, there are no general results characterizing the rate of convergence of FP in the important class of multi-agent games known as \emph{potential games} \cite{Mond96}.

In a potential game, there exists some underlying potential function (the structure of which may be unknown to agents) that all players implicitly seek to optimize. Potential games are particularly important in the study of decentralized control algorithms \cite{marden2009cooperative} and have a wide range of applications in the field of multi-agent systems \cite{scutari2006potential,xu2013decision,zhu2013distributed,ding2012collaborative,
nie2006adaptive,chu2012cooperative,srivastava2005using,Lambert01,
garcia2000fictitious,marden-connections,li2013designing}.

The main contribution of this paper is to show that the rate of convergence of FP is almost always exponential in potential games. In particular, we will show that in almost every potential game and for almost every initial condition, the rate of convergence of FP is exponential (see Theorem \ref{thrm_FP_conv_rate1} for a precise statement of our main result).\footnote{We note that in this paper we study the autonomous variant of continuous-time FP, e.g., \cite{leslie2006generalised,benaim2005stochastic}. Convergence rate estimates for the non-autonomous variant of continuous-time FP can be derived from these estimates using a time change \cite{harris1998rate}.}\footnote{When we say that a property holds for ``almost every'' potential game (initial condition), we mean that the set of potential games (initial conditions) where the property fails to hold has Lebesgue measure zero. See Sections \ref{sec_FP_prelims} and \ref{sec_reg_pot_games} for more details.}

%

We note that there are some fundamental challenges that arise when trying to establish convergence rate estimates for FP in potential games. As pointed out by Harris \cite{harris1998rate}, if (the trajectory of) an FP process in a potential game intersects with a mixed-strategy equilibrium, then it may rest there for an indeterminate amount of time before moving elsewhere. Consequently, solutions of FP in potential games may be non-unique and it is impossible to establish general convergence rate estimates in such games.

In order to overcome these difficulties, the proof of our main result relies crucially on several recent advancements in the study of FP and potential games.
The notion of a regular Nash equilibrium was introduced by Harsanyi \cite{harsanyi1973oddness}. Such equilibria are relatively simple to analyze and posses a variety of useful robustness properties. A game is said to be regular if all equilibria in the game are regular.
While not all potential games are regular, it has recently been shown that almost every potential game is regular \cite{swenson2017regular}.
The work \cite{swenson2017FP_pot_games} studied convergence properties of FP in regular potential games. In particular, in \cite{swenson2017FP_pot_games} it was shown that (i) in a regular potential game, FP converges to a pure strategy NE from almost every initial condition,\footnote{In particular, in any regular potential game, the set of initial conditions from which a mixed-strategy equilibrium can be reached by an FP process has Lebesgue measure zero.} and (ii) in any regular potential game, solutions of FP are unique for almost every initial condition.


We will prove Theorem \ref{thrm_FP_conv_rate1} by studying FP in regular potential games and leveraging these two properties.

We remark that conjectures have been made regarding the rate of convergence of FP in potential games, though no rigorous proofs have yet been given. In particular, Harris conjectured that the rate of convergence of FP is exponential in any \emph{weighted} potential game (\hspace{-.001em}\cite{harris1998rate}, Conjecture 25). Our results partially resolve this conjecture. We show that Harris's conjecture holds for almost every initial condition in almost every \emph{exact} potential game.\footnote{Exact potential games are a subset of weighted potential games \cite{Mond96}. Unless otherwise specified, when referring to a ``potential game'' throughout the paper, we mean an exact potential game.} Our results are less broad than Harris's conjecture in the sense that we consider a smaller class of games (exact potential games) and only prove the result for almost every initial condition. However, our result also makes a stronger statement than encountered in Harris's conjecture in the sense that Harris assumed the constants in the convergence rate estimate would need to be path dependent. We show that for almost every initial condition, the constant in the convergence rate estimate is uniquely determined by the initial condition. See Remark \ref{remark_convergence_conjecture} for more details.

The remainder of the paper is devoted to proving Theorem \ref{thrm_FP_conv_rate1}. Section \ref{sec_prelims} sets up notation and formally defines the FP learning process. Section \ref{sec_reg_pot_games} introduces regular potential games. Section \ref{sec_FP_reg_pot_games} reviews relevant results regarding FP in regular potential games. Section \ref{sec_main_result} states our main result (Theorem \ref{thrm_FP_conv_rate1}) and proves the same. Section \ref{sec_conclusions} concludes the paper.

\section{Preliminaries} \label{sec_prelims}
\subsection{Notation}
A game in normal form is represented by the tuple \newline $\Gamma := (N,(Y_i,u_i)_{i=1,\ldots,N})$, where $N\in\{2,3,\ldots\}$ denotes the number of players, $Y_i=\{y_i^1,\ldots,y_i^{K_i}\}$ denotes the set of pure strategies (or actions) available to player $i$, with cardinality $K_i := |Y_i|$, and $u_i:\prod_{j=1}^N Y_j \rightarrow \mbb{R}$ denotes the utility function of player $i$. Denote by $Y:= \prod_{i=1}^N Y_i$ the set of joint pure strategies, and let $K := \prod_{i=1}^N K_i$ denote the number of joint pure strategies.

For a finite set $S$, let $\triangle(S)$ denote the set of probability distributions over $S$. For $i=1,\ldots,N$, let $\Delta_{i}:=\triangle(Y_i)$ denote the set of \emph{mixed-strategies} available to player $i$. Let $\Delta := \prod_{i=1}^N \Delta_i$ denote the set of joint mixed strategies.\footnote{It is implicitly assumed that players' mixed strategies are independent, i.e., players do not coordinate.} Let $\Delta_{-i} := \prod_{j\in\{1,\ldots,N\}\backslash\{i\}} \Delta_j$. When convenient, given a mixed strategy $\sigma=(\sigma_1,\ldots,\sigma_N)\in \Delta$, we use the notation $\sigma_{-i}$ to denote the tuple $(\sigma_j)_{j\not=i}$.

Given a mixed strategy $\sigma\in\Delta$, the expected utility of player $i$ is given by
$$
U_i(\sigma_1,\ldots,\sigma_N) = \sum_{y\in Y} u_i(y)\sigma_1(y_1)\cdots \sigma_N(y_N).
$$

For $\sigma_{-i} \in \Delta_{-i}$, the best response of player $i$ is given by the set-valued function $\BR_i:\Delta_{-i}\rightrightarrows\Delta_i$,
$$
\BR_i(\sigma_{-i}):= \arg\max_{\sigma_i' \in \Delta_i} U_i(\sigma_i',\sigma_{-i}),
$$
and for $\sigma\in \Delta$ the joint best response is given by the set valued function $\BR:\Delta\rightrightarrows\Delta$
$$
\BR(\sigma) := \BR_{1}(\sigma_{-1})\times\cdots\times \BR_{N}(\sigma_{-N}).
$$

A strategy $\sigma\in \Delta$ is said to be a Nash equilibrium (NE) if $\sigma \in \BR(\sigma)$. For convenience, we sometimes refer to a Nash equilibrium simply as an equilibrium.

We say that $\Gamma$ is a potential game \cite{Mond96} if there exists a function $u:Y\rightarrow \R$ such that $u_i(y_i',y_{-i}) - u_i(y_i'',y_{-i}) = u(y_i',y_{-i}) - u(y_i'',y_{-i})$ for all $y_{-i} \in Y_{-i}$ and $y_i',y_i'' \in Y_i$, for all $i=1,\ldots,N$.

Let $U:\Delta\rightarrow \R$ be the multilinear extension of $u$ defined by
\begin{equation} \label{def_potential_fun1}
U(\sigma_1,\ldots,\sigma_N) = \sum_{y\in Y} u(y)\sigma_1(y_1)\cdots\sigma(y_N).
\end{equation}
The function $U$ may be seen as giving the expected value of $u$ under the mixed strategy $\sigma$.


Using the definitions of $U_i$ and $U$, it is straightforward to verify that
$$
\BR_i(\sigma_{-i}) := \arg\max_{\sigma_i\in\Delta_i} U_i(\sigma_i,\sigma_{-i}) = \arg\max_{\sigma_i\in\Delta_i} U(\sigma_i,\sigma_{-i}).
$$
Thus, in order to compute the best response set, we only require knowledge of the potential function $U$, not necessarily the individual utility functions $(U_i)_{i=1,\ldots,N}$.


By way of notation, given a pure strategy $y_i \in Y_i$ and a mixed strategy $\sigma_{-i} \in \Delta_{-i}$, we will write $U(y_i,\sigma_{-i})$ to indicate the value of $U$ when player $i$ uses a mixed strategy placing all weight on the $y_i$ and the remaining players use the strategy $\sigma_{-i}\in \Delta_{-i}$.

Given a $\sigma_i \in \Delta_i$, let $\sigma_i^k$ denote the value of the $k$-th entry in $\sigma_i$, so that $\sigma_i = (\sigma_i^k)_{k=1}^{K_i}$.
Since the potential function is linear in each $\sigma_i$, if we fix any $i=1,\ldots,N$ we may express it as
\begin{equation} \label{eq_potential_expanded_form}
U(\sigma) = \sum_{k=1}^{K_i} \sigma_i^k U(y_i^k,\sigma_{-i}).
\end{equation}


Throughout the paper we will use the following nomenclature to describe equilibrium strategies.
\begin{mydef}
(i) If an equilibrium strategy $\sigma\in \Delta$ places all its mass on a single action tuple $y=(y_1,\ldots,y_N)\in Y$, then we refer to $\sigma$ as \emph{a pure-strategy equilibrium}.\\
(ii) If an equilibrium strategy $\sigma\in \Delta$ is not a pure-strategy equilibrium, then we say it is a mixed-strategy equilibrium.
\end{mydef}

The following definition gives a refinement of the Nash equilibrium concept applicable to pure-strategy equilibria.
\begin{mydef}
If $\sigma\in \Delta$ is a pure-strategy equilibrium placing mass on the action tuple $y =(y_1,\ldots,y_N) \in Y$, and there holds
$$
u(y_i,y_{-i}) > u(y_i',y_{-i})
$$
for all $y_i'\in Y_i$, $y_i'\not= y_i$, $i=1,\ldots N$, then we say that $\sigma$ is a \emph{strict} pure-strategy Nash equilibrium.
\end{mydef}

If an equilibrium strategy is not pure, then it cannot be strict.
The following definition gives a relaxation of the notion of strictness that applies more generally to mixed-strategy equilibria.
\begin{mydef}
We say that an equilibrium strategy $\sigma\in \Delta$ is \emph{quasi-strict} if for every $i=1,\ldots,N$, every pure strategy $y_i\in BR_i(\sigma_{-i})$ is also in the support of $\sigma_i$.
\end{mydef}
We note that if $\sigma$ is a pure-strategy equilibrium, then $\sigma$ is strict if and only if it is quasi-strict.

In order to study learning dynamics without being (directly) encumbered by the hyperplane constraint inherent in $\Delta_i$ we define
\begin{align*}
X_i := \{x_i\in \R^{K_i-1}:~ 0\leq x_i^k\leq 1 \mbox{ for } k=1,\ldots,K_i-1,\\ \mbox{ and } \sum_{k=1}^{K_i-1}x_i^k \leq 1\},
\end{align*}
where we use the convention that $x_i^k$ denotes the $k$-th entry in $x_i$ so that $x_i = (x_i^k)_{k=1}^{K_i-1}$.

Given $x_i\in X_i$ define the bijective mapping $T_i:X_i\rightarrow \Delta_i$ as
\begin{equation}\label{def_T}
T_i(x_i) = \sigma_i
\end{equation}
for the unique $\sigma_i\in \Delta_i$ such that $\sigma_i^{k} = x_i^{k-1}$ for $k=2,\ldots,K_i$ and $\sigma_i^1 = 1-\sum_{k=1}^{K_i-1} x_i^k$. For $k=1,\ldots,K_i$ let $T_i^k$ be the $k$-th component map of $T_i$ so that $T_i = (T_i^k)_{i=1}^{K_i}$.



Let $X := X_1\times\cdots\times X_N$ and let $T:X\rightarrow \Delta$ be the bijection given by $T = T_1\times\cdots\times T_N$. In an abuse of terminology, we sometimes refer to $X$ as the \emph{mixed-strategy space} of $\Gamma$. When convenient, given an $x\in X$ we use the notation $x_{-i}$ to denote the tuple $(x_j)_{j\not= i}$. Letting $X_{-i} := \prod_{j\not = i} X_j$, we define $T_{-i}:X_{-i} \rightarrow \Delta_{-i}$ as $T_{-i} := (T_j)_{j\not= i}$. Let
\begin{equation}\label{def_kappa}
\kappa := \sum_{i=1}^N (|Y_i|-1)
\end{equation}
denote the dimension of $X$, and note that $\kappa\not= K$, where $K$, defined earlier, is the cardinality of the joint pure strategy set $Y$.

Throughout the paper we often find it convenient to work in $X$ rather than $\Delta$.
In order to keep the notation as simple as possible we overload the definitions given above, modifying all definitions given for strategies $\sigma\in \Delta$ \emph{modus mutandis} using the bijective relationships $T:X\to\Delta$, $T_i:X_i\to\Delta_i$ and $T_{-i}\to\Delta_{-i}$.
In particular, we let $\BR_i:X_{-i} \rightrightarrows X_i$ be defined by
\begin{align*}
\BR_i(x_{-i}) :=  \{ & x_i\in X_i:~ \BR_i(\sigma_{-i}) = \sigma_i,~ \sigma_i\in \Delta_i,\\
& \sigma_{-i} \in \Delta_{-i},~ \sigma_i = T_i(x_i),~ \sigma_{-i} = T_{-i}(x_{-i})\}.
\end{align*}
Similarly, given an $x\in X$ we write $U(x)$ instead of $U(T(x))$, and we
say a strategy $x\in X$ satisfies a given property (e.g., $x$ is a pure-strategy, strict, or quasi-strict equilibrium) if the corresponding strategy $\sigma=T(x)\in \Delta$ satisfies the property.

Given a pure strategy $y_i\in Y_i$, we will write $U(y_i,x_{-i})$ to indicate the value of $U$ when player $i$ uses a mixed strategy placing all weight on the $y_i$ and the remaining players use the strategy $x_{-i}\in X_{-i}$.
Similarly, we will say $y_i^k \in \BR_i(x_{-i})$ if there exists an $x_i\in \BR_i(x_{-i})$ such that $T_i(x_i)$ places weight one on $y_i^k$.

Applying the definition of $T_i$ to \eqref{eq_potential_expanded_form}, we see that $U(x)$ may also be expressed as
\begin{equation} \label{eq_potential_expanded_form2}
U(x) = \sum_{k=1}^{K_i-1}x_i^k U(y_i^{k+1},x_{-i}) + \left(1-\sum_{k=1}^{K_i-1}x_i^k\right)U(y_i^{1},x_{-i}).
\end{equation}
for any $i=1,\ldots,N$.


Other notation as used throughout the paper is as follows:
\begin{itemize}
\item $\N:=\{1,2,\ldots\}$.

\item $\calL^n$, $n\in \{1,2,\ldots\}$ refers to the $n$-dimensional Lebesgue measure.
\end{itemize}

\subsection{Fictitious Play} \label{sec_FP_prelims}
A fictitious play process is formally defined as follows.
\begin{mydef}
An absolutely-continuous mapping $\vx:\R\rightarrow X$ is said to be a \emph{fictitious play process} with initial condition $x_0\in X$ if $\vx(0) = x_0$ and
\begin{equation}\label{def_FP_autonomous}
\dot \vx \in \BR(\vx) - \vx
\end{equation}
holds for almost every $t\in \R$.
\end{mydef}

Given a game $\Gamma$, we say that FP satisfies a property for \emph{almost every} initial condition if the set of initial conditions from which the property fails to hold has $\calL^\kappa$-measure zero.

\section{Regular Potential Games} \label{sec_reg_pot_games}
The notion of a ``regular'' Nash equilibrium was introduced by Harsanyi \cite{harsanyi1973oddness}. Regular equilibria have been studied extensively in the literature (see \cite{van1991stability} and references therein), and have been shown to possess a wide range of desirable stability and robustness properties. Among other things, regular equilibria are quasi-strict \cite{harsanyi1973oddness,van1991stability};  perfect \cite{selten1975reexamination}; proper \cite{myerson1978refinements}; strongly stable \cite{kojima1985strongly}; essential \cite{wen1962essential}; and isolated \cite{van1991stability}.\footnote{For a simple flowchart demonstrating the interrelationships between these concepts, the reader may refer to the survey diagrams found in the appendix of \cite{van1991stability}.}

In addition to their robustness properties, regular equilibria possess an inherently simple analytic structure, which can be quite useful in the study of game-theoretic learning algorithms.
For example, the simple structural properties of regular equilibria were used in a critical way to facilitate the study of FP learning dynamics in \cite{swenson2017FP_pot_games}.

A game is said to be regular if all equilibria in the game are regular.  In this paper we will focus our study on potential games that are regular.

The set of potential games is isomorphic to $\R^{K_p}$, where
$K_p := (\sum_{i=1}^N \prod_{j\not=i} K_j) + N + K-1$, \cite{swenson2017regular}.
We say that almost every potential games possesses a certain property if the set of games where the property fails to hold has $\calL^{K_p}$-measure zero.
Regular potential games were studied in \cite{swenson2017regular}, where the following result was proved.
\begin{theorem}[ \hspace{-.01em}\cite{swenson2017regular}, Theorem 1]
Almost every potential game is regular.
\end{theorem}

The following lemma gives a useful property of pure-strategy equilibria in regular potential games that will be useful in the proof of our main result.
\begin{lemma}\label{lemma_strictness}
Let $x^*\in X$ be a pure-strategy equilibrium of a regular potential game. Then for all $x\in X$ in a neighborhood of $x^*$ there holds
\begin{equation}\label{lemma1_eq0}
\BR(x) = \{x^*\};
\end{equation}
that is, the pure-strategy equilibrium $x^*$ is the unique best response to every $x$ in a neighborhood of $x^*$.
\end{lemma}

\begin{proof}
Without loss of generality, assume that the strategy set $Y_i := \{y_i^1,\ldots,y_i^{K_i}\}$ of each player $i=1,\ldots,N$ is reordered so that $y_i^1 \in \BR_i(x_{-i}^*)$.

In \cite{van1991stability} it is shown that every regular equilibrium is quasi-strict. In particular, this implies that the regular pure-strategy equilibrium $x^*$ is strict, i.e.,
\begin{equation}\label{lemma1_eq3}
U(y_i^1,x_{-i}^*) > U(y_i^{k+1},x_{-i}^*),
\end{equation}
for all $i=1,\ldots,N$, $k=1,\ldots,K_i$.
Differentiating \eqref{eq_potential_expanded_form2} we see that
$$
\frac{\partial U(x^*)}{\partial x_i^k} = U(y_i^{k+1},x_{-i}^*) - U(y_i^{1},x_{-i}^*) < 0.
$$
for all $i=1,\ldots,N$, $k=1,\ldots,K_i$, where the inequality follows from \eqref{lemma1_eq3}.
Since the gradient of $U$ is continuous, we see that
$$
\frac{\partial U(x)}{\partial x_i^k} = U(y_i^{k+1},x_{-i}) - U(y_i^{1},x_{-i}) < 0,
$$
for all $i=1,\ldots,N$, $k=1,\ldots,K_i$, and all $x\in X$ in a neighborhood of $x^*$.
In particular, this gives $U(y_i^{1},x_{-i}) > U(y_i^{k+1},x_{-i})$ for all $i=1,\ldots,N$, $k=1,\ldots,K_i$, and all $x\in X$ in a neighborhood of $x^*$, which is the desired result.


\end{proof}

\section{Fictitious Play in Regular Potential Games} \label{sec_FP_reg_pot_games}
It is known that FP can converge to mixed (but not pure) equilibria in potential games \cite{Mond01}. This can be problematic for several reasons. For one thing, mixed-strategy equilibria necessarily occur at saddle points of the potential function, which means that they do not maximize the potential function and tend to be inherently unstable under learning dynamics. At a more fundamental level, mixed equilibria are problematic since an FP process may reach such an equilibrium in finite time.\footnote{In fact, in \cite{swenson2017FP_pot_games} it was shown that if an FP process converges to a regular mixed equilibrium of a potential game, then the FP process necessarily converges to the equilibrium in finite time. See \cite{swenson2017FP_pot_games}, Section 5.2.} In such an event, the FP process may rest at the equilibrium for an indeterminate amount of time before moving elsewhere. This results in non-uniqueness of solutions and makes it impossible to establish general convergence rate estimates for FP in potential games.

In \cite{swenson2017FP_pot_games} it was shown that these issues can be sidestepped by focusing on the class of regular potential games. In particular, we have the following results for FP in potential games.

\begin{theorem} \label{thrm_main_result_FP_pure_strategy}
Let $\Gamma$ be a regular potential game. Then,\\
(i) For almost every initial condition,
FP converges to a pure-strategy Nash equilibrium. In particular, FP can only reach mixed-strategy (non-pure) equilibria from a set of initial conditions with $\calL^\kappa$-measure zero.\\
(ii) For almost every initial condition $x_0\in X$, there is a unique FP process with $\vx(0) = x_0$.
\end{theorem}

The first item is a restatement of Theorem 1 in \cite{swenson2017FP_pot_games}. The second item follows from Remark 20 in \cite{swenson2017FP_pot_games}.
These properties will play a critical role in the proof of Theorem \ref{thrm_FP_conv_rate1}.


\section{Main Result} \label{sec_main_result}
\begin{theorem} \label{thrm_FP_conv_rate1}
Let $\Gamma$ be a regular potential game. Then for almost every initial condition $x_0\in X$, there exists a constant $c=c(\Gamma,x_0)$ such that if $\vx$ is an FP process associated with $\Gamma$ and $\vx(0) = x_0$, then
\begin{equation} \label{eq_FP_conv_rate1}
d(\vx(t),NE) \leq ce^{- t}.
\end{equation}
\end{theorem}
We note that the constant $c$ in the above theorem is uniquely determined by the game $\Gamma$ and the initial condition $x_0$.
We now prove Theorem \ref{thrm_FP_conv_rate1}.
\begin{proof}
Properties (i) and (ii) of Theorem \ref{thrm_main_result_FP_pure_strategy} imply that
there exists a set $\Omega \subset X$ satisfying the following properties: (a) $\calL^\kappa(X\backslash \Omega)=0$, (b) for every FP process $\vx$ with initial condition $x_0\in\Omega$, $\vx$ is the unique FP process satisfying $\vx(0) = x_0$, and $\vx$ converges to a pure-strategy NE.

Let $x_0\in \Omega$, let $\vx$ be an FP process with $\vx(0) = x_0$, and let $x^*$ be the pure-strategy NE to which $\vx$ converges.
Without loss of generality, assume that the pure-strategy set $Y$ is reordered so that
\begin{equation}\label{eq_x_is_zero}
x^*=0
\end{equation}
(i.e., $T_i^1(x^*_i) = 1$ for all $i=1,\ldots,N$, where $T_i^k$ is defined following \eqref{def_T}).

By Lemma \ref{lemma_strictness}, for all $x$ in a neighborhood of $x^*$ we have $\BR(x) = x^*$. Since $\vx(t) \rightarrow x^*$, this, along with \eqref{def_FP_autonomous} and \eqref{eq_x_is_zero}, implies that
there exists a time $\tau=\tau(\Gamma,x_0)>0$ such that for all $t\geq \tau$,  we have
$\dot{\vx}(t) = - \vx(t)$.
Hence, for $t\geq \tau$ we have $\|\vx(t)\| = \|\vx(\tau)\|e^{\tau-t}$.
Letting $c := \sup_{t\in [0,\tau]} \|\vx(t)\|e^{\tau}$ 
we get $\|\vx(t)\| \leq ce^{-t}$ for all $t\geq 0$.

\end{proof}

\begin{remark} \label{remark_convergence_conjecture}
Conjecture 25 of \cite{harris1998rate} posited that within the class of weighted potential games, the rate of convergence of any FP process $\vx$ is exponential with coefficient $c$ (cf. \eqref{eq_FP_conv_rate1}) depending on the game $\Gamma$ and the particular FP process $\vx$; i.e., $c=c(\Gamma,\vx)$. In Theorem \ref{thrm_FP_conv_rate1} we showed this is true in almost every exact potential game and for almost every initial condition, and furthermore, we showed that, for almost every initial condition, the constant $c$ in \eqref{eq_FP_conv_rate1} can be determined by the initial condition alone, rather than depending on the full path $\vx$; i.e., $c=c(\Gamma,x_0)$.
\end{remark}

\section{Conclusions} \label{sec_conclusions}
The paper studied fictitious play learning dynamics in continuous time. It was shown that in almost every potential game (i.e., in every regular potential game \cite{swenson2017regular}) the rate of convergence of FP is generically exponential. The proof was facilitated by the fact that FP has been shown to converge generically to pure strategy equilibria in regular potential games \cite{swenson2017FP_pot_games}.



\bibliographystyle{IEEEtran}
\bibliography{myRefs}
\end{document}